\documentclass[runningheads]{llncs}
\usepackage{listings}
\usepackage{amssymb}
\usepackage{amsmath}
\usepackage{cite}
\usepackage{graphicx}
\usepackage{comment}
\usepackage[noline,boxed,ruled,noend,linesnumbered]{algorithm2e}
\usepackage{wrapfig}
\usepackage{vcell}
\usepackage{appendix}
\usepackage[hidelinks]{hyperref}
\usepackage{lmodern}
\usepackage{booktabs}
\usepackage{makecell}

\usepackage[usenames,dvipsnames]{xcolor}

\usepackage{hyperref}
\hypersetup{
    colorlinks=true,
    linkcolor=RedViolet,
    anchorcolor=black,
    citecolor=blue,
    urlcolor=black,
    breaklinks=true
}

\usepackage[capitalise]{cleveref}
\Crefname{equation}{}{}
\Crefname{appendix}{Appx.}{Appxs.}
\Crefname{section}{Sect.}{Sects.}
\Crefname{figure}{Fig.}{Figs.}

\usepackage[inline]{enumitem}
\usepackage{csquotes}
\usepackage{marvosym}
\usepackage{multirow}
\usepackage{color}
\usepackage[T1]{fontenc}
\input{macro}

\newcommand{\asserts}{{\tt pr}}
\newcommand{\peq}{{\tt eq}}

\newcommand{\folfp}{\mathsf{FO[LFP]}}
\newcommand{\toolname}[0]{{\sc AutoProof}}

\newlist{compitem}{itemize}{4}
\setlist[compitem,1]{nolistsep,label=$\bullet$}

\newenvironment{talign*}
 {\let\displaystyle\textstyle\csname align*\endcsname}
 {\endalign}

\makeatletter
\@mparswitchfalse
\makeatother
\normalmarginpar

\begin{document}

\title{Proving Functional Program Equivalence via \\Directed Lemma Synthesis}

\author{Yican Sun\inst{1} \and
Ruyi Ji \inst{1}  \and
Jian Fang \inst{1}\and
Xuanlin Jiang \inst{1}\and\\
Mingshuai Chen \inst{3}\and
Yingfei Xiong \inst{1}\inst{2}$^{\text{(\Letter)}}$
}

\authorrunning{Y.~Sun et al.}
\institute{Key Laboratory of High Confidence Software Technologies (Peking University), Ministry of Education; School of
Computer Science, Peking University, Beijing, China\\ \email{\{sycpku,jiruyi910387714,xiongyf\}@pku.edu.cn}\\\email{\{fangjian,xljiang\}@stu.pku.edu.cn}\and
Zhongguancun Laboratory \and
Zhejiang University, Hangzhou, China\\ \email{m.chen@zju.edu.cn}
}

\maketitle

\begin{abstract}
\vspace{-2em}
Proving equivalence between functional programs is a fundamental problem in program verification, which often amounts to reasoning about \emph{algebraic data types} (ADTs) and compositions of \emph{structural recursions}. Modern theorem provers provide  \emph{structural induction} for such reasoning, but a structural induction on the original theorem is often insufficient for many equivalence theorems.
In such cases, one has to invent a set of lemmas, prove these lemmas by additional induction, and use these lemmas to prove the original theorem. There is, however, a lack of systematic understanding of what lemmas are needed for inductive proofs and how these lemmas can be synthesized automatically. This paper presents \emph{directed lemma synthesis}, an effective approach to automating equivalence proofs by discovering critical lemmas using program synthesis techniques. We first identify two \emph{induction-friendly} forms of propositions that give formal guarantees to the progress of the proof. We then propose two tactics that synthesize and apply lemmas, thereby transforming the proof goal into induction-friendly forms. Both tactics reduce lemma synthesis to a set of independent and typically small program synthesis problems that can be efficiently solved. Experimental results demonstrate the effectiveness of our approach: Compared to state-of-the-art equivalence checkers employing heuristic-based lemma enumeration, directed lemma synthesis saves 95.47\% runtime on average and solves 38 more tasks over an extended version of the standard benchmark set.


\end{abstract}

\lstset{
basicstyle=\fontsize{7.5pt}{0.8\baselineskip}\ttfamily,
keywordstyle=\rm\bfseries,
upquote=true,
commentstyle=,
escapeinside=``,
morekeywords={Let, rec, match, with, end, Fixpoint, Inductive, if, then, else, in},
belowskip=0.1em
}

\section{Introduction}
\label{sec:intro}

Automatically proving the equivalence between functional programs is a fundamental problem in program verification. On the one hand, it is the basic way to certify the correctness of optimizing functional programs. On the other hand, since modern theorem provers such as Isabelle~\cite{IsabelleTool}, Coq~\cite{CoqTool}, and Lean~\cite{LeanTool} are based on functional programming languages, many other verification problems reduce to reasoning about equivalence between functional programs.

The core of functional programming languages is built upon \emph{algebraic data types} (ADTs).
An ADT describes composite data structures by combining simpler types; it can be recursive when referring to itself in its own definition.
ADTs are often processed by \emph{structural recursions}, where recursive calls are invoked over the recursive substructures of the input value.
As a result, the crux of verifying functional program equivalence is to reason about the \emph{equivalence between composed structural recursions}, as demonstrated by the following example.

\begin{figure}
\vspace{-2em}
\begin{minipage}[t]{0.5\linewidth}
\begin{small}
\begin{lstlisting}
Inductive List = nil | cons Int List;

Let rev (l:List) =
match l with
| nil `$\rightarrow$` nil
| cons h t `$\rightarrow$` snoc h (rev t)
end;

Let sort (l:List) =
match l with
| nil `$\rightarrow$` nil
| cons h t `$\rightarrow$` ins h (sort t)
end;

Let sum (l:List) =
match y with
| nil `$\rightarrow$` 0
| cons h t `$\rightarrow$` h + (sum t)
end;
\end{lstlisting}
\end{small}
\end{minipage}
\begin{minipage}[t]{0.48\linewidth}
\begin{lstlisting}


Let snoc (x:Int) (l:List) =
match l with
| nil `$\rightarrow$` cons x nil
| cons h t `$\rightarrow$` cons h (snoc x t)
end;

Let ins (x:Int) (l:List) =
match l with
| nil `$\rightarrow$` cons x nil
| cons h t `$\rightarrow$`
    if x `$\le$` h then cons x l
    else cons h (ins x t)
end;
\end{lstlisting}
\end{minipage}
\caption{An algebraic data type and structurally recursive functions.}
\label{fig:spec}
\vspace{-2em}
\end{figure}

\begin{example}\label{example:list-sum}
\Cref{fig:spec} depicts a common ADT \texttt{List} with two constructors, \texttt{nil} and \texttt{cons}, and standard structurally recursive functions, \texttt{rev} that reverses a list, \texttt{sort} that applies insertion sort, and \texttt{sum} that calculates the sum of a list. Functions \texttt{snoc} and \texttt{ins} are for implementing these functions. We are interested in proving that summing a list after reverse is the equivalent of summing a list after sorting:
\begin{small}\begin{align}
    \forall\ {\tt xs: List}.\quad \texttt{sum (rev xs)} \eeq \texttt{sum (sort xs)}~. \tag{$\dag$} \label{eq:prop}
\end{align}\end{small}
To prove the equivalence, it is natural to apply \emph{structural induction}, which has been integrated into modern theorem provers. 
A structural induction certifies that proposition $P(x)$ holds for every instance $x$ of some ADT by showing that $P(x)$ holds for each possible constructor of $x$, assuming the \emph{induction hypothesis} that $P(x')$ holds for the substructure $x'$ of $x$.
For example, a structural induction for \cref{eq:prop} requires to prove two subgoals, each corresponds to a constructor of {\tt List}. The first subgoal is to show \cref{eq:prop} holds when {\tt xs = nil}. The second subgoal induces the following inductive hypothesis.
\begin{small}\begin{align}
    \texttt{sum (rev t)} \eeq \texttt{sum (sort t)}~. \tag{IH} \label{eq:IH2}
\end{align}\end{small}
Proposition \cref{eq:prop} holds for the $\tt cons$ case if: \cref{eq:prop} is true, assuming $\tt xs = cons~h~t$ and \eqref{eq:IH2}.
\qedT
\end{example}

\paragraph{\bf Challenge: Lemma Finding}
Nonetheless, \emph{many theorems cannot be proved by only induction over the original theorem}~\cite{Rippling}. \cref{example:list-sum} is such a case: Its proof requires induction, but induction over~\eqref{eq:prop} is insufficient since we cannot apply the inductive hypothesis~\eqref{eq:IH2}; see~\Cref{app:intro} for a formal proof.
To apply \eqref{eq:IH2}, we have to transform \eqref{eq:prop} until there is a subterm matching either the left-hand-side (LHS) or right-hand-side (RHS) of \eqref{eq:IH2}, such that we can apply \eqref{eq:IH2} to rewrite the transformed formula. However, such a subterm can never be derived through a deductive transformation; see detailed discussion in \cref{sec:motivating}. 

In such cases, it is necessary to invent a set of lemmas, prove these lemmas by additional induction, and use these lemmas to prove the original proposition. Accordingly, the proof process boils down to
\begin{enumerate*}[label=(\roman*)]
    \item \emph{lemma finding}, and
    \item \emph{deductive reasoning with the aid of lemmas}.
\end{enumerate*}
Whereas decision procedures for deductive reasoning have been extensively studied~\cite{CalculusOfComputation, NelsonOppen, Z3}, \emph{there is still a lack of systematic understanding of what lemmas are needed for inductive proofs and how these lemmas can be synthesized automatically}.

Due to the lack of theoretical understanding, many existing automatic proof approaches resort to \emph{heuristic-based lemma enumeration}~\cite{HipSpec,ZipperpositionInd,SupIndMath,kunkakautograde,ImandraSystem,vampireind,cvc4ind,thesy,zeno}. These approaches typically work as follows:
\begin{enumerate*}[label=(\roman*)]
    \item use heuristics to rank all possible lemma candidates in a syntactic space (the heuristics are commonly based on certain machine-learning models or the textual similarity to the original proposition),
    \item enumerate the candidates by rank and
    \item try to prove each lemma candidate and certify the original proposition using the lemma.
\end{enumerate*}
Since there is no guarantee that the lemma candidates are helpful in advancing the proof, such solvers may waste time trying useless candidates, thus leading to inefficiency.
For \cref{example:list-sum}, the enumeration-based solver {\sc HipSpec}~\cite{HipSpec} produces lemma $\tt \forall xs.\ rev~(rev~xs) = xs$, which provides little help to the proof.

\paragraph{\bf Approach}
We present \emph{directed lemma synthesis} to avoid enumerating useless lemmas. From \cref{example:list-sum}, we can see that the key to the inductive proof lies in the \emph{effective application} of the inductive hypothesis. Based on this observation, we identify two syntactic forms 
of propositions that guarantee the effective application of the inductive hypothesis, termed \emph{induction-friendly forms}.
Next, we propose two tactics that synthesize and apply lemmas. The lemmas synthesized by our tactics take the form of an equation, with one of its sides matching a term in the original proposition, and can be used to transform the original proposition by rewriting the matched term into the other side of the lemma. Consequently, the current proof goal splits into two subgoals -- one for proving the transformed proposition and the other for proving the synthesized lemma itself.
Our tactics have the following properties:
\begin{itemize}[label=$\bullet$,leftmargin=*]
\item \emph{Progress}:
The new proof goals after applying our tactics eventually fall into one of the induction-friendly forms. That is, compared with existing directionless lemma enumeration, our synthesis procedure is \emph{directed}: it eventually produces subgoals that admit effective applications of the inductive hypothesis.
\item \emph{Efficiency}: 
The lemma synthesis problem in our tactics can be reduced to a set of independent and typically small \emph{program synthesis} problems, thereby allowing an off-the-shelf program synthesizer to efficiently solve the problems.
\end{itemize}

Based on the two tactics, we propose \toolname{}, an automated approach to proving the equivalence between functional programs by \emph{combining any existing decision procedure with our two tactics for directed lemma synthesis}.

For \cref{example:list-sum}, \toolname{} synthesizes the lemma
\begin{small}\begin{align}
\forall\ \tt xs: List.\quad \texttt{sum (rev xs)} \eeq \texttt{sum xs}~, \tag{L1}\label{eq:L1}
\end{align}\end{small}
where the LHS matches the LHS of the original proposition \eqref{eq:prop}. Therefore, we can use \cref{eq:L1} to rewrite \eqref{eq:prop} into
\begin{small}\begin{align}
\forall\ \tt xs: List.\quad \texttt{sum xs} \eeq \texttt{sum (sort xs)}~. \tag{L2}\label{eq:rem}
\end{align}\end{small}
As will be shown later, both \eqref{eq:L1} and \eqref{eq:rem} fall into the first induction-friendly form, thus ensuring the application of the inductive hypothesis.

\paragraph{Evaluation}
We have implemented \toolname{} on top of {\sc Cvc4Ind}~\cite{cvc4ind} -- the available state-of-the-art equivalence checker with heuristic-based lemma enumeration. We conduct experiments on the program equivalence subset of an extended version of the standard benchmark in automated inductive reasoning. The results show that, compared with the original {\sc Cvc4Ind}, our directed lemma synthesis saves 95.47\% runtime on average and help solve 38 more tasks.

\smallskip
\paragraph{\bf Contributions} The main contributions of this paper include the follows.
\begin{itemize}[leftmargin=*,label=$\bullet$]
    \item The idea of \emph{directed lemma synthesis}, i.e., synthesizing lemmas to transform the proof goal into desired forms.
    \item Two \emph{induction-friendly forms} that guarantee the effective application of the inductive hypothesis, as well as two \emph{tactics} that synthesize and apply lemmas to transform the proof goal into these forms. 
    The lemma synthesis in our tactics can be reduced to a set of independent and typically small synthesis problems, ensuring the efficiency of the lemma synthesis.
    \item The implementation and evaluation of our approach, demonstrating the effectiveness of our approach in synthesizing lemmas to improve the state-of-the-art decision procedures. 
\end{itemize}

\section{Motivation and Approach Overview}
\label{sec:motivating}
In this section, we illustrate \toolname{} over examples. For simplicity, we consider only structurally recursive functions with one parameter in this section.

\paragraph{\bf A Warm-up Example}
To begin with, let us first consider an equation where the direct structural induction yields an effective application of the inductive hypothesis.
\begin{small}\begin{align}
    \forall {\tt xs: List}.\quad \texttt{sum (rev xs)} \eeq \texttt{sum xs} \tag{$\dag_W$} \label{eq:ep}
\end{align}\end{small}
To prove this equation, we conduct a structural induction on $\tt xs$, the ADT argument that the structural recursion traverses, resulting in two cases $\tt xs = nil$ and $\tt xs = cons~h~t$. The first case is trivial, 
and in the second case, we have an inductive hypothesis over the tail list $\tt t$.
\begin{small}\begin{align}
    \texttt{sum (rev t)} \eeq \texttt{sum t} \tag{IH$_W$}\label{eq:IH}
\end{align}\end{small}
We first use the equation $\tt xs = cons~h~t$ to rewrite the original proposition \eqref{eq:ep}, and obtain the following equation.
\begin{small}\begin{align*}
    \texttt{sum (rev (cons h t))} \eeq \texttt{sum (cons h t)}
\end{align*}\end{small}
Here $\tt sum$ and $\tt rev$ are both structural recursions, which use pattern matching to choose different branches based on the constructor of $\tt xs$. With $\tt xs$ replaced as $\tt cons~h~t$,  we can now proceed with the pattern matching and obtain the following equation.
\begin{small}\begin{align}
    \texttt{sum (snoc h (rev t))} \eeq \texttt{h + (sum t)} \label{eq:inter}
\end{align}\end{small}
Now the equation contains a subterm $\tt sum~t$ that matches the RHS of the inductive hypothesis~\eqref{eq:IH}, which allows us to rewrite this equation with \eqref{eq:IH}, resulting in the following equation.
\begin{small}\begin{align}
    \texttt{sum (snoc h (rev t))} \eeq \texttt{h + (sum (rev t))} \label{eq:pre-gen}
\end{align}\end{small}
There is a common ``{\tt rev t}'' term on both sides of the equation above, and we can apply the standard generalization technique to replace it with a new fresh variable {\tt r}, obtaining the following equation.
\begin{small}\begin{align}
    \texttt{sum (snoc h r)} \eeq \texttt{h + (sum r)} \label{eq:post-gen}
\end{align}\end{small}
This equation is simpler than the original one as $\tt snoc$ does not involve calls to other structurally recursive functions. By further applying induction on $\tt r$, we can prove this equation.

We can see that the above proof contains two key steps:
\begin{enumerate*}[label=(\roman*)]
    \item using the inductive hypothesis to rewrite the equation, and
    \item using generalization to eliminate a common non-leaf subprogram.
\end{enumerate*}
We call such two steps an \emph{effective application} of the inductive hypothesis. Note that an effective application is guaranteed because the RHS of the original equation is a single structural recursion call, $\tt sum~xs$. Since a structural recursion applies itself to the substructure of the input, $\tt sum~t$ is guaranteed to appear after reduction. Then, we can use the inductive hypothesis to rewrite, and the rewritten RHS contains $\tt rev~t$. Similarly, the inner-most function call, $\tt rev~xs$, is guaranteed to reduce to $\tt rev~t$. Therefore, a generalization is guaranteed.

\paragraph{\bf Induction-friendly forms} In general, we identify \emph{induction-friendly} forms, where for every equation in this form, there exists a variable such that performing induction on it yields an effective application of the inductive hypothesis for 
the cases involving a recursive substructure.
From the discussion above, we have the simplified version of the first induction-friendly form.

\begin{enumerate}[label=\textbf{(F0)},ref=(F0),nolistsep,leftmargin=*]
    \item\label{form-informal} \emph{(Simplified \labelcref{form1}).} One side of the equation is a single call to a structurally recursive function.
\end{enumerate}

\paragraph{\bf A Harder Example}
\label{sec:harder}
Now let us consider the example equation we have seen in the introduction.
\begin{small}\begin{align}
    \forall\ {\tt xs: List}.\quad \texttt{sum (rev xs)} \eeq \texttt{sum (sort xs)} \tag{$\dag$} \label{eq:prop}
\end{align}\end{small}
Since neither side of \eqref{eq:prop} is a single call to a structurally recursive function, this equation does not fall into \labelcref{form-informal}, 
and indeed, the induction over it will get stuck. To see this point, let us still consider the $\tt x = cons~h~t$ case, where the inductive hypothesis is as follows.
\begin{small}\begin{align}
    \texttt{sum (rev t)} \eeq \texttt{sum (sort t)} \tag{IH} \label{eq:IH2}
\end{align}\end{small}
By rewriting and reducing the original proposition with $\tt x = cons~h~t$, we get the following equation.
\begin{small}\begin{align*}
    \texttt{sum (snoc h (rev t))} \eeq \texttt{sum (ins h (sort t))}
\end{align*}\end{small}
Unfortunately, neither side of \eqref{eq:IH2} appears, disabling the application of the inductive hypothesis. 
In fact, we can formally prove that this proposition cannot be proved by only induction over the original proposition. Please see \Cref{app:intro} for details.

If we can transform the original proposition \eqref{eq:prop} into \labelcref{form-informal}, we can ensure to effectively apply the inductive hypothesis. One way to perform this transformation is to find an equation where one side of the equation is the same as one side of the original proposition, and the other side is a single call to a structurally recursive function. This leads to the lemma \eqref{eq:L1}, which we have seen in the introduction.
\begin{small}\begin{align}
\forall\ \tt xs: List.\quad \texttt{sum (rev xs)} \eeq \texttt{sum xs} \tag{L1}\label{eq:L1}
\end{align}\end{small}
Rewriting \eqref{eq:prop} with \eqref{eq:L1}, we obtain \eqref{eq:rem} we have seen.
\begin{small}\begin{align}
\forall\ \tt xs: List.\quad \texttt{sum xs} \eeq \texttt{sum (sort xs)}\tag{L2}\label{eq:rem}
\end{align}\end{small}
Now the original proof goal \eqref{eq:prop} splits into \eqref{eq:L1} and \eqref{eq:rem}, both conforming to \labelcref{form-informal}. Now we have the guarantee that the inductive hypothesis can be applied in the inductive proofs of both \eqref{eq:L1} and \eqref{eq:rem}.

\paragraph{\bf Automation}
Most steps of the above transformation process can be easily automated, and the only difficult step is to find a suitable lemma. Based on the form of the lemma, the key is finding the structurally recursive function $\tt sum$ to be used on the RHS, equivalent to a known term $\tt sum \circ rev$ on the LHS. In general, synthesizing a function from scratch may be difficult. However, synthesizing a structural recursion is significantly easier for the following two reasons. First, the template fixes a large fraction of codes in a structural recursion. In this example, the structural recursion over {\tt xs} with the following template.
\begin{lstlisting}[xleftmargin=60pt]
    Let f xs =
      match xs with
      | nil `$\rightarrow$` `$base$`
      | cons h t `$\rightarrow$` Let r = f t in `$comb$` h r
      end;
\end{lstlisting}
where the only unknown parts are $base$ and $comb$. Second, we can separate the expression for each constructor as an independent synthesis task. In this example, we have the following two independent synthesis tasks for the constructors $\tt nil$ and $\tt cons$, respectively.
\begin{small}\begin{align*}
    \textstyle\texttt{sum~(rev~nil)} & \eeq base\\
    \textstyle\forall~\texttt{h\ t}.\quad \texttt{sum~(rev~(cons~h~t))} & \eeq comb\ \ {\tt h\ (sum\ (rev\ t))}
\end{align*}\end{small}
Existing program synthesizers~(e.g., {\sc AutoLifter}\cite{autolifter}~in our implementation) can easily solve both tasks. We get $base = 0$ and $comb~h~r=h+r$. Thus, $\tt f$ coincides with ${\tt sum}$.
An additional benefit is that a typical synthesizer requires a verifier to verify the synthesis result. Here, we can omit the verifier and rely on tests to validate the result. This does not affect the soundness of our approach since the synthesized lemma is proved recursively.

\paragraph{\bf Tactic}
Summarizing the above process, we obtain the first tactic. Given a proof goal that does not conform to \labelcref{form-informal}, this tactic splits it into two proof goals, both conforming to \labelcref{form-informal}. This tactic has two variants, which rewrite the LHS and the RHS, respectively. We give only the RHS version here.
In more detail, given an equation $\forall \bar{x}. p_1(\bar{x}) = p_2(\bar{x})$ that does not satisfy \labelcref{form-informal}, our first tactic proceeds as follows.
\begin{enumerate}[leftmargin=*, label = \textit{Step \arabic*.}, ref = \textrm{Step \arabic*},nolistsep]
\item Derive a lemma template in the form of $\forall \bar{x}, p_2(\bar{x}) = f(\bar{x}),$ where $f$ is a structurally recursive function to be synthesized.
\item Generate a set of synthesis problems and solve them to obtain $f$.
\item Generate two proof goals, $\forall \bar{x}. p_1(\bar{x}) = f(\bar{x})$ and $\forall \bar{x}. f(\bar{x}) = p_2(\bar{x})$.
\end{enumerate}

\paragraph{\bf Overall Process}
Our approach \toolname{} combines any deductive solver with the two tactics to prove equivalence between functional programs. Given an equation, our approach first invokes the deductive solver to prove the equation. If the deductive solver fails to prove, we check if the equation is in an induction-friendly form, and apply induction to generate new proof goals. Otherwise, we check if any tactic can be applied, and apply the tactic to generate new proof goals. Finally, we recursively invoke our approach to the new proof goals.

\paragraph{\bf Towards the Full Approach}
The tactic we present here attempts to transform a complex term into a single structural recursion, but it may not be possible in general. Thus, the full tactic transforms only a composition of two structural recursions into a single one each time, to significantly increase the chance of synthesis success.

Through out the section we consider only structurally recursive functions taking only one parameter, but
there may be multiple ADT variables in general (e.g., proving the commutativity of natural number multiplications). 
Our second tactic deals with an issue caused by \textit{inconsistent recursions}, that is, different recursions that traverse different ADT variables.
Examples and details on this tactic can be found in \Cref{sec:tactic2}.

\section{Preliminary}
\label{sec:problem}
This section presents the background of program equivalence checking. We first articulate the range of equivalence checking tasks. Throughout this paper, we use $p(v_1,\ldots,v_k)$ to denote a functional program $p$ whose free variables range from $\{v_1,\ldots,v_k\}$.

\paragraph{\bf Types} The family of types in \toolname{} consists of two disjoint parts: (1) the algebraic data types, and  (ADTs)~\cite{SoftwareFoundations}, and
(2) the built-in types such as {\tt Int} or {\tt Bool}. For ease of presentation, we assume that there is only one built-in type {\tt Int} for integers, and only one ADT for lists with integer elements. {\tt List} has two constructors, {\tt nil: List} for the empty list, and {\tt cons: Int $\rightarrow$ List $\rightarrow$ List} that appends an integer at the head of a list. \toolname{} can be easily extended to handle all ADTs and more built-in types.

\paragraph{\bf Syntax} As illustrated in \Cref{fig:syntax}, the specification for an equivalence checking task is generated by {\sf SPEC}, where
each task consists of two parts.

First, a specification defines a sequence of \emph{canonical} structural recursions (CSRs), each generated by {\sf CSRDef}. A CSR $f$ is a function whose last argument is of an ADT.
It applies pattern matching to the last argument $v_k$, which we call the \emph{recursive argument}, and considers all top-level constructors of $v_k$. If $v_k = {\tt nil}$, i.e., an empty list, it invokes {\it base}$(v_1,\ldots,v_{k-1})$ generated by {\sf PROG}. Otherwise, $v_k = {\tt cons\ h\ t}$. It recursively invokes itself over $\tt t$ with all other arguments {unchanged}, stores the result of the recursive call in $\tt r$, and then combines the result via the program $comb(v_1\ldots v_{k-1}, {\tt h}, {\tt r})$ generated by {\sf PROG}. The non-terminal {\sf PROG} generates either a variable $var$, a numerical constant $constant$, or an application by
\begin{enumerate*}[label=(\arabic*)]
    \item a built-in operator $op$ for a built-in type (e.g., $+,-,\times$ for {\tt Int}),
    \item a constructor $ctr$ of an ADT, and
    \item a CSR $f$,
\end{enumerate*}
followed with $k$ programs, where $k$ is the number of arguments required by this application.

Having defined all CSRs, a specification gives the equation $\forall{\bar{x}}. p_1(\bar{x}) = p_2(\bar{x})$, where $p_1$ and $p_2$ are generated by {\sf PROG}.

\paragraph{\bf Semantics} We adapt standard evaluation rules~\cite{CoqTool} to the syntax~(\Cref{fig:syntax}), see \Cref{app:prelims} for details. We use \emph{term reduction} to refer to a single-step evaluation. 

\paragraph{\bf Abstraction} An \emph{abstraction} is a syntactic transformation from a program $p$ to another program $p'$ performed in steps. In each step, given a program $p$, it introduces a fresh variable and replaces a subprogram of $p$ with the fresh variable. For example, we can abstract the program $p$ of {\tt sum (snoc (h + h) (rev t))} to $p'$ of {\tt sum (snoc a b)}, which replaces {\tt (h + h)} to {\tt a}, and {\tt (rev t)} to {\tt b}.

Note that if $p'$ is an abstraction of $p$, any transformation on $p'$ yields another transformation on $p$ by simply replacing each introduced fresh variable back with the corresponding subprogram. For example, the transformation from $p'$ to {\tt a + (sum b)} yields the transformation from $p$ to {\tt (h + h) + sum (rev t)}.

\begin{small}
\begin{figure}[t]
\begin{center}
\begin{tabular}{rcl}
{\sf SPEC} &::= & {\sf CSRDef}$^*$ $\forall {\bar x}. p_1 = p_2$\\
& where & $p_1, p_2\in {\sf PROG}$\\
{\sf CSRDef}   &::= &
{\tt\bf Let} $f$ $v_1$ $v_2$ $\ldots$ $v_k$ = \\
&&
{\tt\bf match} $v_k$ {\tt\bf with}  \\
&&
 {\tt | nil} $\rightarrow$\ \  $base$    \\
&&
 {\tt | cons h t} $\rightarrow$  {\tt{\bf Let} r }$=\ f$\ \  $v_1$ $\ldots$ $v_{k- 1}$ \ {\tt t} {\tt\bf in} $comb$\\
& & {\bf end};\\
            &   where &  $base, comb\in {\sf PROG}$\\ 
{\sf PROG} &::= &{\it f}\ \ {\sf PROG$^*$}
                 $\mid$ {\it ctr}\ \ {\sf PROG$^*$}
                 $\mid$ {\it op}\ \ {\sf PROG$^*$}
                 $\mid$ {\it var}
                 $\mid$ {\it const}
                 \\
            &    where& $f$ is a CSR, $ctr$ is a constructor of ADT, \\
            & &  $const$ is a constant with the built-in type,\\
            &    & $op$ is a primitive operator, and $var$ is a free variable.\\
\end{tabular}
\end{center}
\vspace{-1em}
\caption{Syntax of the surface language of \toolname.}
\label{fig:syntax}
\end{figure}
\end{small}

\paragraph{\bf Expressivity}
Compared with widely-considered structural recursions~\cite{CoqTool}, CSR introduces two additional restrictions. First, it applies pattern-matching to only one argument. Second, it keeps other parameters unchanged in recursive calls. However, we can always transform any structural recursion into a composition of CSRs by a refinement of \emph{defunctionalization}~\cite{defunctionalization}. Thus, restricting SRs to CSRs does not affect the expressivity of functional programs, see \Cref{app:prelims} for details.

\section{\toolname{} in Detail}
\label{sec:tactics}

\subsection{The Overall Approach}
\label{sec:integrate}

\begin{figure}[t]
\begin{minipage}[t]{0.5\linewidth}
\begin{footnotesize}
\begin{lstlisting}[language=python,escapeinside=``,tabsize=2,stepnumber=1,numbers=left,keywordstyle=\bfseries\color{purple},basicstyle=\fontsize{7.5pt}{0.8\baselineskip}\ttfamily,xleftmargin=20pt]
class lemma_tactic:
  # to be instantiated
  def precond(eq): pass
  def extract(eq): pass

  def t_apply(eq):
    `$p_s'$`,`$v$` = extract(eq)
    lem = syn_lem(`$p_s'$`,`$v$`)
    eq`$'$` = apply_lem(eq, lem)
    return eq`$'$`, lem
\end{lstlisting}
\end{footnotesize}
\end{minipage}
\begin{minipage}[t]{0.48\linewidth}
\begin{lstlisting}[language=python,escapeinside=``,tabsize=2,stepnumber=1,numbers=left,keywordstyle=\bfseries\color{purple},firstnumber=11,,basicstyle=\fontsize{7.5pt}{0.8\baselineskip}\ttfamily,morekeywords={then}]
# tactics: set of built-in
           tactics
def Prove(pr,eq):
  if try_deductive(pr,eq) succeeds:
    return
  else:
    if induction-friendly(eq) then:
      subgoals = split(pr,eq)
      for sg in subgoals: Prove(sg)
      return
    for t in tactics:
      if t.precond(eq) then:
        eq`$'$`,lem = t.t_apply(eq)
        Prove(pr,lem)
        Prove(pr.append(lem),eq`$'$`)
        return
\end{lstlisting}
\end{minipage}
\caption{Pseudocode of \toolname{}}
\label{fig:code}
\end{figure}

The pseudo-code of \toolname{} is shown in \Cref{fig:code}.
The main procedure is {\tt Prove} (Lines 11--24). The input of this procedure is a pair $(\asserts, \peq)$, termed as a \textit{goal}, where $\asserts$ is short for {premises}, which is a set of equations including all lemmas and inductive hypotheses, and $\peq$ is an equation denoting the current proposition to be proved. The target of a goal is to prove $\asserts\vdash \peq$.

{\tt Prove} wraps an underlying deductive solver responsible for performing standard deductive reasoning, such as reduction or applying a premise. {\tt Prove} first invokes the deductive solver to prove the input goal (Line 12).
If the deductive solver succeeds, the proof procedure finishes~(Lines 13--14). \toolname{} is compatible with any deductive solver. We choose the deductive reasoning module of the state-of-the-art solver {\sc Cvc4Ind}~\cite{cvc4ind} in our implementation.

Otherwise, the goal is too complex for the deductive solver to handle, which often requires finding a lemma. In this case, \toolname{} first invokes {\tt induction\\-friendly(e)} to check if the input equation $\peq$ satisfies one of the two identified forms~\labelcref{form1} and \labelcref{form2}~(defined in \Cref{sec:indfriend-forms}). If so, then by the properties of induction-friendly forms, the original goal can be split into a set of subgoals~(Line 18) by induction with effective applications of the inductive hypotheses.

If not, \toolname{} applies a built-in set 
of tactics
to transform an input equation into an induction-friendly form gradually. We will discuss tactics in detail in~\Cref{sec:routine}. A tactic generally has a precondition, i.e., {\tt precond($\cdot$)} indicating the set of applicable equations. If the tactic is applicable~(Line 21), \toolname{} invokes another procedure {\tt t\_apply} that synthesizes a lemma {\tt lem} and applies this lemma to transform the input equation $\peq$ into another equation $\peq'$. (Line 22). Then, {\tt Prove} is recursively called to prove the lemma ${\tt lem}$ and the equation $\peq'$ with the aid of {\tt lem} (Lines 24--25).

In this algorithm, induction is applied only when the proof goal is in the induction-friendly form, hence we need a \emph{progress} property that,
starting from any goal, if all lemmas are successfully synthesized, the initial goal can be eventually transformed into an induction-friendly form. This property is formally proved in \Cref{thm:progress}.

\subsection{Induction-friendly Forms in \toolname{}}
\label{sec:indfriend-forms}
\toolname{} identifies two induction-friendly forms (defined at \Cref{sec:motivating}). Both forms guarantee the effective application of the inductive hypothesis.

\begin{enumerate}[label=\textbf{(F1)}, ref=(F1),leftmargin=*,nolistsep]
    \item\label{form1} The first induction-friendly form is $f\ v_1\ldots\ v_k = p(v_1,\ldots,v_k)$, where
    \begin{enumerate}[label=\textbf{(F1.\arabic*)},ref=(F1.\arabic*),nolistsep]
    \item\label{form11} One side of the equation is in the form {\it f $v_1\ldots v_k$}, where $f$ is a CSR and $v_1\ldots v_k$ are different. From the definition of CSR, $f$ applies pattern-matching on $v_k$.
    \item\label{form12} The other side of the equation is a program $p(v_1\ldots v_k)$ satisfies the condition as follows. If $v_k$ appears in $p$, then there exists an occurrence of $v_k$, such that (1) $v_k$ appears as the recursive argument of the CSR it is passed to, and (2) all other arguments in this CSR invocation do not contain $v_k$.
    \end{enumerate}
\end{enumerate}

\begin{figure}
\vspace{-2em}
\begin{minipage}{0.53\linewidth}
\begin{lstlisting}
Let app x y =
match y with
| nil `$\rightarrow$` nil
| cons h t `$\rightarrow$`
    cons h (app x t)
end;
\end{lstlisting}
\end{minipage}
\begin{minipage}{0.41\linewidth}
\begin{lstlisting}
Let sapp x y z =
match z with
| nil `$\rightarrow$` (sum x) + (sum y)
| cons h t `$\rightarrow$` h + (sapp x y z)
end;
\end{lstlisting}
\end{minipage}
\caption{More CSRs for This Section}
\label{fig:app}
\vspace{-2em}
\end{figure}

Intuitively, \labelcref{form11} guarantees the applicability of the inductive hypothesis, and \labelcref{form12} guarantees that there is a common term for generalization. To be more concrete, consider proving {\tt $\tt \forall x,y,z.$\ sapp x y z = sum (app (app y z) x)}, where $\tt app$ and {\tt sapp} are defined in~\Cref{fig:app}, {\tt app} is the list concatenation function, and {\tt sapp} calucates the sum of three concatenated lists. Note that this equation fulfills \labelcref{form1}. Induction over {\tt z} and consider the {\tt cons} case where {\tt z = cons h t}, the LHS can be reduced to:
\small{$$\tt h + (sapp\ x\ y\ t) \eeq sum\ (app\ (cons\ h\ (app\ y\ t))\ x)$$}
Due to \labelcref{form11}, the LHS contains a single call, and due to the definition of the CSR, the recursive call must take $\tt t$ as the recursive argument and keep the other argument unchanged. Therefore, the LHS must contain
$\tt sapp\ x\ y\ t$ as a subprogram, making the induction hypothesis applicable. Applying the induction hypothesis, we get
\small{$$\tt h + (sum\ (app\ (app\ y\ t)\ x)) \eeq (app\ (cons\ h\ (app\ y\ t))\ x)$$}
Due to \labelcref{form12}, either {\tt z} do not appear in RHS, leading to exactly the same RHS as the inductive hypothesis, or we can find an occurrence of $\tt z$ in the RHS ({\tt app y z} in this example), such that {\tt z} is the recursive argument and all other arguments do not contain {\tt z}. In this case, the reduction produces the recursive call {\tt app y t}, a common subprogram on both sides. In both cases, we can generalize this subprogram to a fresh variable, yielding an effective application.

The second form is dedicated to our tactics. We propose this form to capture the lemmas proposed by our second tactic~(\Cref{sec:tactic2}).
\begin{enumerate}[label=\textbf{(F2)}, ref=(F2),leftmargin=*,nolistsep]
    \item\label{form2} The second form is
    $f\ v_1\ \ldots\ v_k \eeq f'\ v_1'\ldots\  v_k'$, where $v_i\neq v_j \wedge v'_i \neq v'_j$ for all $1\le i<j\le k$,
    i.e., each side is a single CSR call whose arguments are distinct variables.
    
\end{enumerate}
When the equation fulfills \labelcref{form2}, we can guarantee an effective application of the induction hypothesis by a nested induction over $v_k$ and $v_k'$. For example, consider proving $\tt \forall x, y, z.$ {\tt sapp x y z = sapp x z y}. We first perform induction over $\tt z$ and consider the {\tt cons} case where {\tt z = cons ${\tt h_1}$ ${\tt t_1}$}, the goal reduces to the following equation with the hypothesis $\tt sapp\ x\ y\ t_1 = sapp\ x\ t_1\ y$.
\begin{small}$$\tt h_1 + sapp\ x\ y\ t_1 = sapp\ x\  (cons\ h_1\ t_1)\ y$$\end{small}
Applying the hypothesis on LHS, we obtain the following subgoal:
\begin{small}$$\tt h_1 + sapp\ x\ t_1\ y = sapp\ x\  (cons\ h_1\ t_1)\ y$$\end{small}
Note that this subgoal falls into \labelcref{form1}, where the RHS is a single call and $\tt y$ is only used as a recursive argument, and thus an effective application of inductive hypothesis is guaranteed when we perform induction over $\tt y$. We can see that this conformance to \labelcref{form1} is guaranteed because the single call on the LHS guarantees the application of the inductive hypothesis, which will make the recursive arguments on both sides the same.

\smallskip
The following theorem establishes that both \labelcref{form1} and \labelcref{form2} are induction-friendly. The proof is deferred to \Cref{app:tactics}.
\begin{theorem}
\label{thm:friendly}
    Both \labelcref{form1} and \labelcref{form2} are induction-friendly.
\end{theorem}

\subsection{General Routine of Tactics}
\label{sec:routine}

In this part, we demonstrate the general routine of how tactics are applied to transform the input goal, i.e., the {\tt t.t\_apply($\cdot$)} function in Line~6 of \Cref{fig:code}. Let us start with the notation of \emph{abstraction}.

\paragraph{\bf Tactics} 
Informally, our tactics focus on lemmas that transform a fragment of the input equation into a single CSR invocation.
Thus, it requires a subroutine {\tt extract($\cdot$)}, which needs to be instantiated per tactic, to extract the specification of a lemma synthesis problem from the equation to be proved. The output of {\tt extract($\cdot$)} is a tuple $(p_s', v)$, where $p_s'$ is an abstraction of the subprogram to be transformed, and $v$ is a free variable in $p_s'$ (Line 7 in \Cref{fig:code}). The output $(p_s', v)$ indicates the following lemma synthesis problem.
\begin{equation}
\label{eq:lemma-formal}
        \forall \tilde{v}. \forall v.  f^*\ {\tilde v}\ v \eeq p_s'(\tilde v, v) \tag{$\peq_1$}
\end{equation}
where $\tilde{v}$ is the set of all free variables other than $v$.

The approach to finding $f^*$ has been fully presented in \Cref{sec:motivating} and thus is omitted here. As long as the program synthesis succeeds in finding $f^*$, we propose the lemma \eqref{eq:lemma-formal} above. Since $p_s'$ is an abstraction of some subprogram in the input equation, we can easily apply the lemma~\eqref{eq:lemma-formal} to transform the input equation and obtain a new equation $\peq_2$ to be proved~(Lines 8--9 in \Cref{fig:code}).

\subsection{Tactic 1: Removing Compositions}
\label{sec:tactic1}

Our first tactic is used to guarantee \labelcref{form11}. Thus, the precondition {\tt t.precond(eq)} returns true if {\tt eq} does not satisfy \labelcref{form11}. Below, we demonstrate the extract function in detail.

The extract function picks a non-leaf subprogram $c\ p_1\ p_2\ \ldots\ p_k$ of some side of the input equation {\tt eq}, where $c$ is a primitive operator, a constructor, or a CSR, $p_1\ldots p_k$ are the arguments of $c$, and at least one of $p_i$ is not a variable.  Then, we abstract all arguments passed to each $p_i$ with a fresh variable, obtaining the abstracted subprogram $p_s'$. We define the cost of this extraction as the number of fresh variables introduced. The extraction returns the extraction with the minimum cost. If there are several choices with the same minimum cost, we pick an arbitrary one.

For example, consider proving the equation $\texttt{app (rev a) (rev (rev b))} = \texttt{rev}\\ \texttt{(rev (app (rev a) b))}$, where {\tt app} is the list concatenation function presented in \Cref{fig:app}. Then, we may choose the subprogram {\tt rev (rev (app (rev a) b))} and abstract the argument {\tt app (rev a) b} of the inner {\tt rev} with a fresh variable {\tt x}, obtaining $p_s' = $ {\tt rev (rev x)}. Since this extraction only introduces one variable, the cost is one, which is the minimum cost.

Having fixed $p_s'$, we then select a variable $v$ in $p_s'$ to be the recursive argument of the synthesized CSR $f^*$. We choose the variable whose corresponding lemma fulfills the maximum number of forms in~\labelcref{form11}, \labelcref{form12}, and \labelcref{form2}. If there is a tie, we choose an arbitrary variable that reaches the maximum. Note that the lemma generated by this tactic satisfies at least~\labelcref{form11}, which guarantees the applicability of the inductive hypothesis.

\subsection{Tactic 2: Switching Recursive Arguments}
\label{sec:tactic2}
Our second tactic is used to guarantee \labelcref{form12}, and synthesizes a lemma such as {\tt f x y = f$'$ y x} to switch the recursive argument of a function (recall that the recursive argument is always the last one). This tactic is only invoked when the first tactic~(\Cref{sec:tactic1}) cannot apply. Thus, the precondition {\tt precond(eq)} returns true if {\tt eq} satisfies \labelcref{form11} but not \labelcref{form12}. Without loss of generality, we assume the LHS is a single CSR invocation with the recursive argument $x$.

The extraction algorithm picks the occurrence of $x$ with the maximum depth in the AST, where $x$ is passed to a CSR $f$. Then, each $p_i$ is either the variable $x$ or a program that does not contain $x$ (otherwise, we find an occurrence of $x$ with a larger depth). We introduce fresh variables $v_1\ldots v_k$ to abstract $p_1\ldots p_k$. For some $1\le i<k$ such that $p_i = x$ (such $i$ always exists since the equation violates~\labelcref{form12}), the extract outputs $p_s' = f\ v_1\ \ldots\ v_k$ and $x = v_i$.
Since all arguments of $f$ are abstracted, the lemma proposed by this tactic must satisfy \labelcref{form2}. As a result, the lemma is induction-friendly.

For example, consider proving $\forall${\tt x, y, z. plus3 y z x = plus (plus x y) z}. Note that this equation satisfies \labelcref{form11} but not \labelcref{form12}. We choose the subprogram {\tt plus x y} and abstract it into $p_s' = $ {\tt plus a b}. Note that {\tt x} appears as the first argument, thus the algorithm outputs $(p_s',{\tt a})$, which requires to synthesize a lemma $\forall${\tt a, b. plus a b = plus' b a}. As long as the lemma is synthesized, we can replace {\tt plus x y} to {\tt plus' y x}, making the equation satisfying~\labelcref{form12}.

\subsection{Properties}

First, we show the soundness of \toolname{}, which is straightforward.
\begin{theorem}[Soundness]
    If \toolname{} proves an input goal, then the goal is true.
\end{theorem}
\begin{proof}
    The proof of the input equation searched by \toolname{} is a sequence of induction, reduction, and application of lemmas. Thus, the soundness of \toolname{} follows from the soundness of these standard tactics.
\end{proof}

\paragraph{\bf Progress} 
As mentioned in \Cref{sec:integrate}, the effectiveness of \toolname{} comes from the following progress theorem, where the proof is deferred to \Cref{app:tactics}. 

\begin{theorem}[Progress]
\label{thm:progress}
    Starting from any goal, if all lemmas are successfully synthesized, the initial goal can be eventually transformed into an induction-friendly form.
\end{theorem}

\section{Evaluation}
\label{sec:evaluation}

We implement \toolname{} on top of {\sc Cvc4Ind}~\cite{cvc4ind}, an extension of {\sc Cvc4} with induction and the available\footnote{{\sc Pirate}~\cite{Pirate} is reported to have better performance than {\sc Cvc4Ind} on \emph{standard benchmarks} in our evaluation, but its code and its experimental data are not publicly accessible. Thus, we do not compare our approach against {\sc Pirate}. Note that \toolname{} can be combined with any deductive solver, including {\sc Pirate}.} state-of-the-art prover for proving equivalence between functional programs. We choose {\sc AutoLifter}~\cite{autolifter} as the underlying synthesizer, which can solve the synthesis tasks in \Cref{sec:routine} over randomly generated tests.
{\sc Cvc4Ind} comes with a lemma enumeration module, our implementation invokes only the deductive reasoning module of {\sc Cvc4Ind}. 
To compare the lemma enumeration with directed lemma synthesis, we evaluate \toolname{} against {\sc Cvc4Ind}.

\begin{small}
\begin{table}[t]
\setlength{\tabcolsep}{5pt}
\begin{center}
\caption{Experimental results over the benchmarks.}
\label{tab:exp}
\begin{tabular}{cccccc}
\toprule
 & \makecell{\#Solved\\(Standard)} & \makecell{\#Solved\\(Extension)}  & \makecell{\#Solved\\(Total)}& \makecell{\#Fails\\ (Timeout) } &\makecell{AvgTime}\\
\midrule
\toolname{}     &
\makecell{\textbf{140}\\($\bf \uparrow16.67\%$) }
& \makecell{\textbf{21}\\($\bf \uparrow600\%$)} & \makecell{\textbf{161}\\($\bf \uparrow30.89\%$)} & \textbf{109}    & \makecell{\textbf{3.64s}\\($\bf \downarrow95.47\%$)} \\
\midrule
\sc{Cvc4Ind} &  120 & 3 & 123 & 147   & 80.36s \\
\bottomrule
\end{tabular}
\end{center}
\vspace*{-2\baselineskip}
\end{table}
\end{small}

\paragraph{\bf {Dataset}}
We collect 248 \emph{standard benchmarks} from the equivalence checking subset of CLAM~\cite{Rippling}, Isaplaaner~\cite{johansson2010case}, and ``Tons of Inductive problems'' (TIP)~\cite{claessen2015tip}, which have been widely employed in previous works~\cite{Rippling,johansson2010case,cvc4ind,ZipperpositionInd,adtind}.
We observe that these benchmarks do not consider the mix of ADTs and other theories (e.g., LIA for integer manipulation), which is also an important fragment in practice~\cite{LIA1,LIA2,LIA3,LIA4,LIA5}.
Thus, we created 22 \emph{additional benchmarks} combining the theory of ADTs and LIA by converting ADTs to primitive types in existing benchmarks, such as converting {\tt Nat} to {\tt Int}. Our test suite thus consists of 270 benchmarks in total.

\paragraph{\bf {Procedure}} We use our implementation and the baseline to prove the problems in the benchmarks. We set the time limit as 360 seconds for solving an individual benchmark, the default timeout of {\sc Cvc4Ind} and is aligned with previous work~\cite{cvc4ind,vampireind,ZipperpositionInd,adtind}. We obtain all results on the server with the Intel(R) Xeon(R) Platinum 8369HC CPU, 8GB RAM, and the Ubuntu 22.04.2 system.

\paragraph{\bf {Results}}
The comparison results are summarized in Table \ref{tab:exp}.
Overall, \toolname{} solves 161 benchmarks, while the baseline {\sc Cvc4Ind} solves 123, showing that directed lemma synthesis can make an enhancement with a ratio of 30.89\%.
On the solved benchmarks, \toolname{} takes 3.64s on average, while {\sc Cvc4Ind} takes 80.36s, indicating that directed lemma synthesis can save 95.47\% runtime.
The results justify our motivation: compared with the directionless lemma enumeration, directed lemma synthesis can avoid wasting time on useless lemmas.
Note that \toolname{} shows significant strength on additional benchmarks with a mixed theory. This is because the tactics and induction-friendly forms in our approach are \emph{purely syntactic}, making \toolname{} \emph{theory-agnostic}. In contrast, {\sc Cvc4Ind} is \emph{theory-dependent}. Thus, it is hard for {\sc Cvc4Ind} to tackle benchmarks with mixed theories.

\paragraph{\bf {Discussion}} We observe that in the failed cases, the failure to synthesize a lemma is a common cause, and this in turn is due to two reasons. The first one is that the program synthesizer fails to produce a solution for a solvable synthesis problem. For example, one equation involves an exponential function, whose implementation is extremely slow on ADT types, and the synthesizer timed out on executing the randomly generated tests. The second one is that the potential lemma requires a structural recursion that is not canonical. Though in theory such a structural recursion can be converted into compositions of CSRs, our current algorithm only supports the synthesis of CSRs, and thus cannot synthesize such lemmas. This observation shows that, if we can further improve program synthesis in future, our approach may prove more theorems.

\section{Related Work}
\label{sec:related}

\paragraph{\bf Lemma Finding in Inductive Reasoning}
Due to the necessity, the lemma finding algorithm has been integrated into various architectures of inductive reasoning, including theory exploration~\cite{thesy,HipSpec}, superposition-based provers~\cite{ZipperpositionInd,SupIndMath,ImandraSystem,vampireind}, SMT solvers~\cite{cvc4ind,adtind,fossil,IC3+}, and other customized approaches~\cite{zeno,kunkakautograde}. These approaches can be divided into two categories.

First, most of these approaches~\cite{HipSpec,ZipperpositionInd,SupIndMath,kunkakautograde,ImandraSystem,vampireind,cvc4ind,thesy,zeno,adtind} apply lemma enumeration based on heuristics or user-provided templates, which often produce lemmas with little help to the proof, leading to inefficiency, as we have discussed in~\Cref{sec:intro}. Compared with these approaches, \toolname{} considers the \emph{directed} lemma synthesis and application, eventually producing subgoals in induction-friendly forms.

Second, there are approaches~\cite{fossil,IC3+} considering the lemma synthesis over a decision procedure based on bounded quantification and pre-fixed point computation. These approaches are restricted to structural recursions without nested function invocations or constructors, which cover only 19/248~(7\%) benchmarks in our test suite~(\Cref{sec:evaluation}).

\paragraph{\bf Other Approaches in Functional Program Verification}
There are other approaches~\cite{dafny,liquidhaskell,leon,fluid} verifying the properties of functional programs \emph{without} induction. These tools require the user to manually provide an induction hypothesis. Thus, these approaches cannot prove any benchmark in our test suite~(\Cref{sec:evaluation}).

\paragraph{\bf Invariant Synthesis}
Lemma synthesis has also been applied to verifying the properties of imperative programs~\cite{nonlinear-invariant,ICE}, where the lemma synthesis is often recognized as \emph{invariant synthesis}. Since the core of imperative programs is the mutable atomic variables and arrays instead of ADTs, previous approaches for invariant synthesis~\cite{nonlinear-invariant,ICE} cannot be applied to our problem. It is future work to understand whether we can extend \toolname{} for verifying imperative programs.

\section{Conclusion}
\label{sec:conclusion}
We have presented \toolname{}, a prover for verifying the equivalence between functional programs, with a novel {directed lemma synthesis} engine. The conceptual novelty of our approach is the induction-friendly forms, which are propositions that give formal guarantees to the progress of the proof. We identified two forms and proposed two tactics that synthesize and apply lemmas, transforming the proof goal into induction-friendly forms. Both tactics reduce lemma synthesis to a specialized class of program synthesis problems with efficient algorithms. We conducted experiments, showing the strength of our approach. In detail, compared to state-of-the-art equivalence checkers employing heuristic-based lemma enumeration, directed lemma synthesis saves 95.47\% runtime on average and solves 38 more tasks over a standard benchmark set.

 \bibliographystyle{splncs04}
 \bibliography{PL.bib}
 \appendix
 \clearpage
\section{Missing Details in~\Cref{sec:intro}}
\label{app:intro}

We prove the following theorem.
\begin{theorem}\label{thm1}
    It is impossible to prove~\eqref{eq:prop} without induction.
\end{theorem}
\begin{proof}
We prove this claim following the model-theoretic idea of the previous work~\cite{modeltheory, Skolem1934, fluid}.
To investigate inductive reasoning without lemmas, we consider a subsystem ${\sf AX}$ of $\folfp$, whose axioms consist of ADTs and RDF, and inference rules are unfolding and applying an axiom.
    \begin{align}
        {\tt nil}&\in {\tt List}\label{app1}\\
        \forall h, t. (h\in {\tt Int}\land t\in {\tt List})&\rightarrow ({\tt cons}\ h\ t\in {\tt List})\label{app2}\\
        {\tt sum}\ {\tt nil} &= 0\label{app3}\\
        \forall h: {\tt Int}, t: {\tt List}.\ {\tt sum}\ ({\tt cons\ h\ t}) &= h + {\tt sum}\ t\label{app4}\\
        {\tt rev}\ {\tt nil} &= {\tt nil}\label{app5}\\
        \forall h: {\tt Int}, t: {\tt List}.\ {\tt rev}\ ({\tt cons\ h\ t}) &= {\tt snoc}\ h\ ({\tt rev}\ t)\label{app6}\\
        {\tt sort}\ {\tt nil} &= {\tt nil}\label{app7}\\
        \forall h: {\tt Int}, t: {\tt List}.\ {\tt rev}\ ({\tt cons\ h\ t}) &= {\tt ins}\ h\ ({\tt sort}\ t)\label{app8}\\
        \forall x: {\tt Int}.\ {\tt snoc}\ x\ ({\tt nil}) &= {\tt cons}\ x\ {\tt nil}\label{app9}\\
        \forall x, h: {\tt Int}, t: {\tt List}.\ {\tt snoc}\ x\ ({\tt cons}\ h\ t) &= {\tt cons}\ h\ ({\tt snoc}\ x\ t)\label{app0}\\
        \forall x: {\tt Int}.\ {\tt ins}\ x\ ({\tt nil}) &= {\tt cons}\ x\ {\tt nil}\label{app11}\\
        \label{app12}\forall x, h: {\tt Int}, t: {\tt List}.\ {\tt ins}\ x\ ({\tt cons}\ h\ t) &=
        \begin{cases}
            {\tt cons}\ x\ l & x\le h\\
            {\tt cons}\ h\ ({\tt ins}\ x\ t) & \neg(x\le h)
        \end{cases}
    \end{align}

Then, we construct models for {\tt List, nil, cons, sum, rev, sort, snoc, ins, xs} so that all axioms~(\eqref{app1}--\eqref{app12}) above hold, but {\tt sum (rev $\tt xs$)} $\ne$ {\tt sum (sort $\tt xs$)}.

We define {\tt List} consists of standard finite lists and all elements $(x,i)$ such that (1) $i\ge 0$, and (2) $x$ is an infinite list and $\forall j\ge i, x_j=0$. {\tt nil} is the empty list as usual. For a non-standard list $t = (x,i)$, we define:
\begin{itemize}
    \item ${\tt cons}\ h\ t$ as $(x', i+1)$ where $x'$ is the output after prepending $h$ into $x$.
    \item {\tt sum}\ $t$ as $\sum_{j=1}^i{x_j}$.
    \item {\tt rev}\ $t$ as $(0\cdots,0)$.
    \item {\tt snoc}\ $h$ $t$ as $t$.
    \item {\tt sort}\ $t$ as $(x',i')$ where $x'_0,\ldots, x'_{i'-1}$ is the sorted non-positive elements of $x_0\ldots,x_{i-1}$.
    \item {\tt ins}\ $h$ $t$ as follows. If $t = (0\cdots,0)$ and $h\le 0$, then {\tt ins}\ $h$ $t$ is $(h0\cdots,1)$. If $t = (0\cdots,0)$ and $h> 0$, then {\tt ins}\ $h$ $t$ is $(0\cdots,0)$. Otherwise, the definition of {\tt ins}\ $h$\ $t$ follows from \Cref{app12}.
\end{itemize}
It is easy to verify that these definitions satisfy all axioms~(\eqref{app1}--\eqref{app12}) above. Define {\tt xs} to be $((-1)0\cdots,1)$, then {\tt sum (rev $\tt xs$)} $=\ 0$ but {\tt sum (sort $\tt xs$)} $=\ -1$. Thus \Cref{thm1} follows.\qed
\end{proof}

\begin{theorem}\label{thm2}
   It is impossible to prove \eqref{eq:prop} without lemmas.
\end{theorem}
\begin{proof}
    Consider the model in the proof of \Cref{thm1}. Note that \eqref{eq:prop} holds for all substructures of ${\tt xs}$, thus \Cref{thm2} follows.
    \qed
\end{proof}

 \section{Missing Details in \Cref{sec:problem}}
\label{app:prelims}
\subsection{Evaluation Rules}
There are three rules for evaluation. We define built-in values to be the constants with the built-in type.
\begin{align*}
    f\ p_1\ p_2\ldots\ p_{k-1}\ {\tt nil} &\rightarrow base(p_1,\ldots, p_{k-1})\\
    f\ p_1\ p_2\ldots\ ({\tt cons}\ p_h\ p_t) &\rightarrow comb(p_1,\ldots, p_{k-1}, p_h, f\ p_1\ldots p_{k-1}\ p_t)\\
    op\ v_1\ v_2\ldots\ v_l &\rightarrow v^*
\end{align*}
where $f$ is a CSR with $k$ arguments, $op$ is a primitive-operator for built-in types with $l$ arguemts, $v_1\ldots v_{l}$ are built-in values, and $v^*$ is the constant computed by the definition of $op$.

\subsection{Expressivity}
We present the proof for lists, which can be generalized to other ADTs. A general structural recursion is in the following form, where there is always an argument $v$ that is decreasing, and other parameters may change along the recursion.

\begin{center}
\begin{tabular}{l}
{\tt\bf Let} $f$ $\bar{v}$ $\ldots$ ($v${\tt : List}) = \\
{\tt\bf match} $v$ {\tt\bf with}  \\
 {\tt | nil} $\rightarrow$\ \  $base$    \\
 {\tt | cons h t} $\rightarrow$  {\tt{\bf Let} r }$=\ f$\ \  $\psi(v,\bar{v})$ \ {\tt t} {\tt\bf in} $comb$\\
{\bf end};\\
\end{tabular}
\end{center}
We define another datatype to record the recursion history.
$$\texttt{Inductive List2 = nil2 [List] | cons2 [List] Int List2};$$
where $\tt [List]$ refers to list of lists.
Then, we define follows functions.

\begin{center}
\begin{tabular}{l}
{\tt\bf Let} $\tt snoc2$ $\tt v$ $\tt a$ = \\
{\tt\bf match} $\tt a$ {\tt\bf with}  \\
 {\tt | []} $\rightarrow$\ \  {\tt [v]}    \\
 {\tt | h3::t3} $\rightarrow$  {\tt{\bf Let} r }$=$ $\tt snoc2$ $\tt v$ $\tt a$ {\tt\bf in}
 {\tt h3::r}
 \\
{\bf end};\\
{\tt\bf Let} $\tt mapsnoc2$ $\tt v$ $\tt r$ = \\
{\tt\bf match} $\tt r$ {\tt\bf with}  \\
 {\tt | nil2 a} $\rightarrow$\ \  {\tt nil2 (snoc2 v a)}    \\
 {\tt | cons2 his h t2} $\rightarrow$  {\tt{\bf Let} r }$=$ $\tt mapsnoc2$ $\tt v$ $\tt t$ {\tt\bf in}
 {\tt cons2 (snoc2 v his) h r}
 \\
{\bf end};\\
{\tt\bf Let} $\tt gen\ v$ = \\
{\tt\bf match} $\tt v$ {\tt\bf with}  \\
 {\tt | nil} $\rightarrow$\ \  {\tt nil2 []}    \\
 {\tt | cons h t} $\rightarrow$  {\tt{\bf Let} r = gen t} {\tt\bf in}
 {\tt cons2 [v] h (mapsnoc2 v r)}
 \\
{\bf end};\\
{\tt\bf Let} $\tt exec2\ \bar{v}\ a$ = \\
{\tt\bf match} $\tt a$ {\tt\bf with}  \\
 {\tt | []} $\rightarrow$\ {$\tt \bar{v}$}    \\
 {\tt | h3::t3} $\rightarrow$ {\tt{\bf Let} r = exec2 $\tt \bar{v}$ t3} {\tt\bf in} $\tt \psi(h3,r)$
 \\
{\tt\bf Let} $\tt exec\ \bar{v}\ genv$ = \\
{\tt\bf match} $\tt genv$ {\tt\bf with}  \\
 {\tt | nil2 a} $\rightarrow$\ \  {\tt{\bf Let} tmp = exec2 $\tt \bar{v}$ a} {\tt\bf in} {\tt base(tmp2)}    \\
 {\tt | cons2 his h t2} $\rightarrow$ {\tt{\bf Let} tmp = exec2 $\tt \bar{v}$ his} {\tt\bf in}\\
 \quad\quad\quad\quad\quad\quad\quad\quad\quad
 {\tt{\bf Let} r = exec $\tt \bar{v}$ t} {\tt\bf in}
 {\tt comb(tmp, h, r)}
 \\
{\bf end};\\
\end{tabular}
\end{center}

Thus, $f\ \bar{v}\ v = {\tt exec}\ \bar{v}\ ({\tt gen}\ v)$.

 \section{Missing Details in \Cref{sec:tactics}}
\label{app:tactics}
\subsection{Proof of \Cref{thm:friendly}}
\begin{proposition}\label{prop1}
\labelcref{form1} is induction-friendly.
\end{proposition}

\begin{proposition}\label{prop2}
\labelcref{form2} is induction-friendly.
\end{proposition}

\Cref{thm:friendly} follows by the two propositions above.

\begin{proof}[of \Cref{prop1}]
    Recall the first induction-friendly form is as follows. $$ f\ v_1\ldots\ v_k = p(v_1,\ldots,v_k)$$
    \begin{enumerate}[label=\textbf{(F1.\arabic*)},ref=(F1.\arabic*),nolistsep,leftmargin=*]
    \item One side of the equation is in the form {\it f $v_1\ldots v_k$}, where $f$ is a CSR and $v_1\ldots v_k$ are different. From the definition of CSR, $f$ applies pattern-matching on $v_k$.
    \item Another side of the equation is a program $p$ with the variables $v_1\ldots v_k$ satisfies the following. For each occurrence of $v_k$, if $v_k$ is passed to the CSR $g$, we need to ensure $v_k$ appears as a recursive argument of $g$.
    \end{enumerate}

    Denote $v_1,\ldots,v_{k-1}$ as $\tilde v$. Induction over $v_k$ and consider the {\tt cons} case. The LHS becomes $comb(\tilde{v},h,f\ \tilde v\ t)$, RHS becomes $p(\tilde v,{\tt cons}\ h\ t)$. Thus, we can apply inductive hypothesis on LHS, and the LHS becomes  $comb(\tilde{v},h,p(\tilde v, t))$. By \labelcref{form12}, either it does not contain $v_k$, leading to exactly the same RHS as the inductive hypothesis, or
    there is a subprogram $f^*\ p_1^*\ldots\ p_{l}^*$ of $p(\tilde v, v_k)$ such that $v_k$ does not appear in $p_1^*, \ldots, p_{l-1}^*$, and $p_l^* = v_k$. In this time, we can use the definition of CSR and obtain $f^*\ p_1^*\ \ldots p_{l-1}^*\ t$, which is a subprogam of $p(\tilde v, t)$, leading to a term for generalization.\qed
\end{proof}

\begin{proof}[of \Cref{prop2}]
    Recall the second induction-friendly form as follows. $$f\ v_1\ \ldots\ v_k \eeq f'\ v_1'\ldots\  v_k'$$
    i.e., each side is a single CSR call whose arguments are distinct variables.

    Without loss of generality, we denote $v_l$ to be $v_k$, $v_r$ to be $v_k'$, and $\tilde v$ to be all other variables. We rephrase the equation above to be
    $$f\ \tilde{v}\ v_r\ v_l \eeq f'\ \tilde{v}\ v_l\ v_r$$
    First induction on $v_l$ and consider the {\tt cons} case, i.e., $v_l = {\tt cons}\ h_l\ t_l$. Apply the inductive hypothesis. We have that.
    $$comb(\tilde v, v_r, h_l, f'\ \tilde v\ t_l\ v_r) = f'\ \tilde v\ ({\tt cons}\ h_l\ t_l)\ v_r$$
    Then induction over $v_r$ and consider the {\tt cons} case, i.e., $v_r = {\tt cons}\ h_r\ t_r$. Apply the inductive hypothesis. We have that.
    \begin{align*}&comb(\tilde v, v_r, h_l, comb'(\tilde v, t_l, h_r, f'\ \tilde v\ t_l\ t_r)) \\& \quad \quad \quad =comb'(\tilde v, {\tt cons}\ h_l\ t_l, h_r, comb(\tilde v, t_r, h_l, f'\ \tilde v\ t_l\ t_r))\end{align*}
    where we can generalize the common term $f'\ \tilde v\ t_l\ t_r$.
\end{proof}

\subsection{Proof of \Cref{thm:progress}}
We first recall the theorem as follows.
\begin{theorem}
\label{thm:prgs2}
Starting from any goal, if all lemmas are successfully synthesized, the initial goal can be eventually transformed into an induction-friendly form.
\end{theorem}

We first prove the following lemma.
\begin{lemma}\label{lem31}
    \Cref{thm:prgs2} holds if the input goal satisfies \labelcref{form11}.
\end{lemma}
\begin{proof}
    Note that we only apply the second tactic~(\Cref{sec:tactic2}) if the input goal satisfies \labelcref{form11}. Let $\phi(e)$ to be the number of occurrences of $v_k$ in $e$ where $v_k$ does not appear as a recursive argument, and let $e_in$ be the input equation. If $\phi(e_in)=0$ then the input equation is induction-friendly, and thus we are done. Otherwise, applying the second tactic, we get two subgoals: (1) the lemma itself, which satisfies \labelcref{form2} and thus is induction-friendly; (2) the equation after applying the lemma, this is an equation $e'$ with $\phi(e') = \phi(e_in) - 1$. Thus, applying a finite number of tactic application, it eventually gets $\phi = 0$ and becomes induction-friendly.
\end{proof}

We are ready to prove \Cref{thm:prgs2}.
\begin{proof}[of \Cref{thm:prgs2}]
If the input-goal satisfies \labelcref{form11}, then we are done due to \Cref{lem31}. Otherwise, let $\psi(e)$ to be the number of compositions in the equation $e$ and let $e_in$ be the input equation. We have that $\psi(e_in)>0$. Applying the first tactic, we get two subgoals: (1) the lemma itself, which satisfies~\labelcref{form11} and thus can be  eventually reduced to induction-friendly forms; (2) the equation after applying the lemma, this is an equation $e'$ with $\psi(e') < \psi(e_in)$. Thus, applying a finite number of tactic application, it eventually gets $\phi = 0$ and becomes induction-friendly.
\end{proof}

\end{document}